\documentclass[journal]{IEEEtran}
\usepackage{amsmath,amssymb,amsfonts}
\usepackage{cite}
\usepackage{graphicx}
\usepackage{bm}
\usepackage{textcomp}
\usepackage{booktabs} 

\newcommand{\dd}{\mathrm{d}}

\newcommand {\Define} {\stackrel {\Delta} {=}  }
\newcommand{\mya}{\mathrel{\overset{\makebox[0pt]{{\tiny(a)}}}{=}}}

\newtheorem{theorem}{Theorem}

\begin{document}

\title{Improving Doppler Resilience of OFDM through Delay-Doppler Sensing}

\author{Danish Nisar, Saif Khan Mohammed, Muhammad Ubadah, Ronny Hadani and Robert Calderbank

\thanks{D. Nisar, S. K. Mohammed, and M. Ubadah are with the Department of Electrical Engineering, Indian Institute of Technology Delhi, India (eez198429@ee.iitd.ac.in, saifkmohammed@gmail.com, eez198356@ee.iitd.ac.in). S. K. Mohammed is currently with Cohere Technologies Inc., CA, USA, on extra-ordinary leave from I.I.T. Delhi.
R. Hadani is with the Department of Mathematics, University of Texas at Austin, USA (hadani@math.utexas.edu) and also with Cohere Technologies Inc., CA, USA. R. Calderbank is with the Department of Electrical and Computer Engineering, Duke University, USA (robert.calderbank@duke.edu).}}

\maketitle

\begin{abstract}
The performance of traditional CP-OFDM degrades severely in doubly-spread wireless channels due to inter-carrier interference (ICI). In this paper, we propose DD domain sensing based CP-OFDM where we transmit a Zadoff-Chu (ZC) pilot signal overlaid on CP-OFDM data carriers. At the receiver, DD domain signal processing is used to acquire the effective DD domain channel filter
which is stationary in the DD domain. From this DD domain estimate, we derive the complete frequency domain (FD) input-output (I/O) relation between CP-OFDM carriers, acquiring which is otherwise difficult with traditional time-frequency signal processing. Using this FD I/O relation, we estimate the received FD pilot signal which is then canceled from the received FD signal, resulting in a data-only signal. Joint detection of all CP-OFDM data carriers from this data-only signal equalizes the effect of ICI. Numerical simulations of the standardized 3GPP TDL-C channel shows that in high mobility scenarios, the proposed DD domain sensing based CP-OFDM achieves significantly better spectral efficiency when compared to that achieved by traditional CP-OFDM.
\end{abstract}


\section{Introduction}
Next generation wireless communication systems are expected to achieve reliable high-throughput communication even in doubly-spread channels \cite{IMT2030, OnTheRoadTo6G}. The performance of existing 4G/5G systems based on CP-OFDM (Cyclic Prefix - Orthogonal Frequency Division Multiplexing) degrades severely in high mobility scenarios, due to inter-carrier interference (ICI) resulting from channel Doppler spread \cite{Nee2000, Wang2006}. Recently Zak-OTFS (Zak - Orthogonal Time Frequency Space) modulation was proposed, whose performance does not degrade even in high mobility scenarios \cite{zakotfs1, zakotfs2, otfsbook}. In Zak-OTFS, information is carried by waveforms which are narrow quasi-periodic pulses in the delay-Doppler (DD) domain. Zak-OTFS can also be implemented over CP-OFDM \cite{Zak_OTFS_over_CP_OFDM}, as a precoder at the transmitter which converts the DD domain information into frequency domain (FD) symbols which are input to a CP-OFDM modulator. Similarly, at the receiver, the FD symbols at the output of the CP-OFDM demodulator are converted to DD domain symbols.
This Zak-OTFS over CP-OFDM architecture allows for the benefits of Zak-OTFS modulation to be achieved in existing 4G/5G modems as the precoder at the transmitter and the post-processor at the receiver can be implemented in software.

In CP-OFDM, we traditionally estimate the channel response on pilot sub-carriers and interpolate them to obtain the channel response on the data sub-carriers. In CP-OFDM it is challenging to acquire the interference coefficient
between sub-carriers. That is, it is difficult to estimate the complete FD input-output (I/O) relation accurately as it changes with time due to the doubly-selective nature of wireless channels. In the absence of the complete I/O relation, even joint equalization of all carriers cannot equalize the effect of ICI.

On the other hand, the Zak-OTFS I/O relation in the DD domain is stationary, i.e., the DD domain output is given by twisted convolution between the effective DD domain channel filter and the DD domain input signal. Due to this stationarity, the interaction between a doubly-spread channel and the Zak-OTFS carrier waveforms is predictable, i.e., the channel response to any Zak-OTFS carrier waveform can be accurately predicted from the response to a particular Zak-OTFS pilot carrier waveform \cite{zakotfs2, otfsbook}.
Further, when the channel delay/Doppler spread is less than the delay/Doppler period of the quasi-periodic Zak-OTFS carrier pulse, the DD domain I/O relation can be accurately and completely acquired using a single Zak-OTFS pilot carrier waveform \cite{zakotfs1, zakotfs2}.

In this paper, we propose DD domain sensing based CP-OFDM, where we leverage the predictability of the Zak-OTFS I/O relation to acquire the effective DD domain channel filter which then gives us the complete FD I/O relation. We then jointly equalize all CP-OFDM sub-carriers which equalizes the effect of ICI, resulting in performance which is robust to channel Doppler spread. To be precise we consider the Zadoff-Chu (ZC) pilot signal which is used in 3GPP 5G NR, and which has good auto-ambiguity function required for accurate channel estimation and has low peak-to-average-power ratio (PAPR) \cite{Mattu_ZC_DD}.
This pilot signal is overlaid on top of the information/data symbols carried by each sub-carrier and therefore we do not need to transmit traditional CP-OFDM pilot carriers. This reduces the pilot overhead compared to traditional CP-OFDM.

At the receiver, we receive a superposition of the pilot and data signal. Leveraging the Zak-OTFS over CP-OFDM architecture, we convert the received FD symbols at the output of the CP-OFDM demodulator to DD domain symbols. Cross-ambiguity between the received DD domain symbols and the ZC pilot signal gives an estimate of the effective DD domain channel filter. Although the received DD domain symbols also contain contribution from the data signal, the cross-ambiguity/de-spreading gain of the ZC pilot signal reduces the contribution of the data signal to the estimated effective DD domain channel filter. This channel filter is then transformed to give us the complete FD I/O relation. Using this estimated I/O relation we compute an estimate of the received FD pilot signal which is then canceled from the received FD symbols resulting in a ``data-only" received signal. We then jointly equalize all sub-carriers in this received FD data-only signal which gives us an estimate of the transmitted FD information symbols. The estimated information symbols and the estimated FD I/O relation are then used to derive an estimate of the received data signal which is then canceled from the received signal resulting in a ``pilot-only" signal.
Cross-ambiguity between this pilot-only signal and the transmitted ZC pilot signal gives a better channel estimate. In this manner, we iterate between channel estimation and data detection resulting in improvement in the channel estimation accuracy and the spectral efficiency with increasing number of iterations.

The proposed DD domain sensing based CP-OFDM is compatible with existing 4G/5G modems as in the transmitter we only need to add a ZC pilot signal on each sub-carrier, and all receiver processing can be performed in software on the received FD symbols at the output of the CP-OFDM demodulator. With $K$ sub-carriers, the proposed DD domain sensing based CP-OFDM receiver has a complexity of $O(K^2)$ compared to the $O(K \, \log(K))$ complexity of traditional CP-OFDM receiver. The higher complexity is due to the joint equalization of all sub-carriers which is the price for achieving robustness towards channel Doppler spread.

Through simulations, we compare the
effective spectral efficiency (SE) achieved by the proposed DD domain sensing based CP-OFDM and traditional CP-OFDM for the 3GPP TDL-C channel model \cite{3gpptr38901}. For both systems, we optimize the SE w.r.t. the allocation of power between pilot and data, and the highest possible modulation and coding scheme (MCS) for which an average low-density parity check code (LDPC) block error rate (BLER) of $0.1$ or lower is achievable. For a given total transmit power to noise ratio (TPNR), for low mobility scenarios (low channel Doppler spread) CP-OFDM achieves better SE than the proposed DD domain sensing based CP-OFDM. However, with increasing channel Doppler spread, the SE achieved by traditional CP-OFDM degrades severely in comparison to that achieved by the proposed DD domain sensing based CP-OFDM. We also study the impact of the pilot to data ratio on the achieved SE. The main contributions of our work in this paper are as follows.

\begin{itemize}
    \item We propose DD domain sensing based CP-OFDM which is compatible with existing 4G/5G modems.
    \item At the transmitter, a ZC pilot signal is overlaid on the CP-OFDM data carriers. This saves overhead of CP-OFDM pilot sub-carriers.
    \item At the receiver, we acquire the effective DD domain channel filter through DD domain signal processing, followed by conversion of the acquired estimate to an estimate of the complete FD I/O relation. This is then used to cancel the pilot signal resulting in a data-only signal which is used to jointly equalize all CP-OFDM data sub-carriers. We iterate between channel estimation and data detection to get performance improvement.
    \item We provide exhaustive numerical study on the variation in the channel estimation accuracy and the SE performance with varying channel Doppler spread, and TPNR. 
\end{itemize}
The paper is organized as follows. In Section \ref{seczakotfs} we review Zak-OTFS when implemented over CP-OFDM. In Section \ref{embda} we present the proposed DD domain sensing based CP-OFDM, followed by numerical simulations in Section \ref{simsec}.

\section{Zak-OTFS modulation}
\label{seczakotfs}
\begin{figure*}[h]
\vspace{-1mm}
\centering
\includegraphics[width=16.8cm, height=6.9cm]{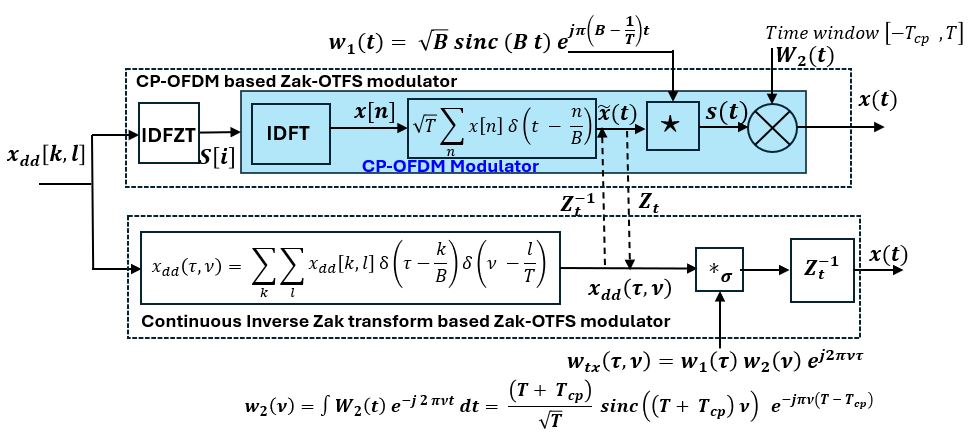}
\vspace{-2mm}
\caption{
Zak-OTFS modulation implemented as IDFZT followed by CP-OFDM modulator (see top chain).}
\label{fig1}
\vspace{-5mm}
\end{figure*}
In Fig.~\ref{fig1}, the top signal processing chain shows the architecture of a Zak-OTFS over OFDM transmitter,
which implements Zak-OTFS modulation as a precoder
over CP-OFDM \cite{Zak_OTFS_over_CP_OFDM}. This enables implementation of Zak-OTFS modulation in existing 4G/5G modems.

Let $B$ be the bandwidth of the communication signal and
$T=1/\Delta f$, where $\Delta f$ is the sub-carrier spacing of the underlying CP-OFDM modulator. Let $\tau_p$ and $\nu_p = 1/\tau_p$ denote the delay and Doppler period of Zak-OTFS modulation. Let $M \Define B \tau_p$, $N \Define T \nu_p$.
Let $x[k,l]$, $k=0,1,\cdots, M-1$, $l=0,1,\cdots, N-1$
denote the $MN = BT$ information symbols. These are firstly
embedded onto the discrete DD domain signal $x_{dd}[k,l]$ given by
\begin{eqnarray}
    x_{dd}[k,l] & = & x[k \, \text{mod} \, M,l \, \text{mod} \, N] \, e^{j 2 \pi \lfloor \frac{k}{M} \rfloor \frac{l}{N}},
\end{eqnarray}$k,l \in {\mathbb Z}$. Note that $x_{dd}[k,l]$ is quasi-periodic with periods
$M$ and $N$ respectively along the discrete delay and Doppler axes, i.e. for all $k,l,n,m \in {\mathbb Z}$
\begin{eqnarray}
\label{qpeqn234}
    x_{dd}[k + nM, l + mN] & = & e^{j 2 \pi n \frac{l}{N}} \, x_{dd}[k,l].
\end{eqnarray}As shown in Fig.~\ref{fig1}, the DD domain signal is converted to the discrete frequency-domain (FD) signal $S[i], i \in {\mathbb Z}$ using the Inverse Discrete Frequency Zak Transform (IDFZT)
    \begin{eqnarray}
\label{sieqn12}
    S[i] & = & \text{IDFZT}(x_{dd}[k,l]) \nonumber \\
    & = & \frac{1}{\sqrt{M}} \sum\limits_{k=0}^{M-1} x_{dd}[k,i] \, e^{-j 2 \pi \frac{i k}{MN}}, \,\,i \in {\mathbb Z}.
\end{eqnarray}Note that $S[i], i \in {\mathbb Z}$ is periodic with period $MN$ because $x_{dd}[k,l]$ is exactly periodic along the discrete Doppler axis with period $M$ (see (\ref{qpeqn234})), i.e., $x_{dd}[k, i + mN] = x_{dd}[k,i]$ and therefore $x_{dd}[k, i + MN] = x_{dd}[k,i]$.
\begin{eqnarray}
    S[i+ MN] & = & S[i],
\end{eqnarray}for all $i \in {\mathbb Z}$. $S[i]$ is then input to a CP-OFDM modulator shown in Fig.~\ref{fig1} whose output is the transmit time-domain signal $x(t)$ which is related to $S[i]$ through the following equations.
$S[i]$ is converted to the discrete-time signal through inverse Discrete Fourier Transform (IDFT)
\begin{eqnarray}
    x[n] & = & \frac{1}{\sqrt{MN}} \sum\limits_{i=0}^{MN-1} S[i] \, e^{j 2 \pi \frac{i n}{MN}},
\end{eqnarray}$n \in {\mathbb Z}$. Note that $x[n]$ is periodic with period $MN$, $x[n + MN] = x[n]$.
Finally
\begin{eqnarray}
    x(t) & = & W_2(t) \, \left[ w_1(t) \star \Tilde{x}(t)\right], \nonumber \\
    \Tilde{x}(t) & \Define & \sqrt{T} \sum\limits_{n \in {\mathbb Z}} x[n] \, \delta\left( t - \frac{n}{B}\right).
\end{eqnarray}

In \cite{Zak_OTFS_over_CP_OFDM}, it has been shown that $x(t)$
is a Zak-OTFS modulated signal with $x_{dd}[k,l]$ as the information signal and $w_{tx}(\tau, \nu) = w_1(\tau) w_2(\nu) \, e^{j 2 \pi \nu \tau}$ as the pulse shaping filter, $w_2(\nu) = \int W_2(t) \, e^{-j 2 \pi \nu t} \, dt$, i.e.
\begin{eqnarray}
    x(t) & = & {\mathcal Z}_t^{-1}\left( w_{tx}(\tau, \nu) *_{\sigma} x_{dd}(\tau, \nu) \right), \nonumber \\
    x_{dd}(\tau, \nu) & \Define & \sum\limits_{k,l \in {\mathbb Z}} x_{dd}[k,l] \, \delta(\tau - k/B) \, \delta(\nu - l/T), \nonumber \\
\end{eqnarray}where ${\mathcal Z}_t^{-1}$ denotes the inverse Zak transform which gives the time-domain representation of any DD domain signal and $*_{\sigma}$
denotes the twisted convolution operation in the DD domain
\begin{eqnarray}
    x_{dd}^{w_{tx}}(\tau, \nu) & \Define & w_{tx}(\tau, \nu) *_{\sigma} x_{dd}(\tau, \nu), \nonumber \\
    &  & \hspace{-27mm} = \iint w_{tx}(\tau', \nu') \, x_{dd}(\tau - \tau', \nu - \nu') \, e^{j 2 \pi \nu' (\tau - \tau'}) \, d\tau' \, d\nu'.
\end{eqnarray}The inverse Zak transform of the DD domain signal $x_{dd}^{w_{tx}}(\tau, \nu)$ is given by
\begin{eqnarray}
\label{eqnxtzak}
    x(t) & = & {\mathcal Z}_t^{-1}\left(  x_{dd}^{w_{tx}}(\tau, \nu) \right) \nonumber \\
    & = & \sqrt{\tau_p} \int\limits_{0}^{\nu_p} x_{dd}^{w_{tx}}(t, \nu) \, d\nu.
\end{eqnarray}Since the CP-OFDM signal also has a cyclic prefix of duration $T_{cp}$ and assuming that $x(t)$ is band-limited to $\left[-\frac{1}{2T} \,,\, B  -\frac{1}{2T} \right]$, we specifically choose the pulse shaping filters to be
\begin{eqnarray}
    w_1(\tau) & = & \sqrt{B} \, \text{sinc}(B \tau) \, e^{j \pi \left( B - \frac{1}{T}\right) t}, \nonumber \\
    w_2(\nu) & = & \frac{(T + T_{cp})}{\sqrt{T}} \, \text{sinc}((T + T_{cp}) \nu) \, e^{-j \pi \nu (T - T_{cp})},
\end{eqnarray}which corresponds to rectangular time-windowing, i.e.
\begin{eqnarray}
\label{eqnW2t}
    W_2(t) & \Define  \begin{cases}
       \frac{1}{\sqrt{T}} &, -T_{cp} \leq t < T, \\
       0  &, \mbox{\small{otherwise}} \\
    \end{cases}.
\end{eqnarray}
\begin{figure*}[h]
\vspace{-2mm}
\centering
\includegraphics[width=16.8cm, height=6.9cm]{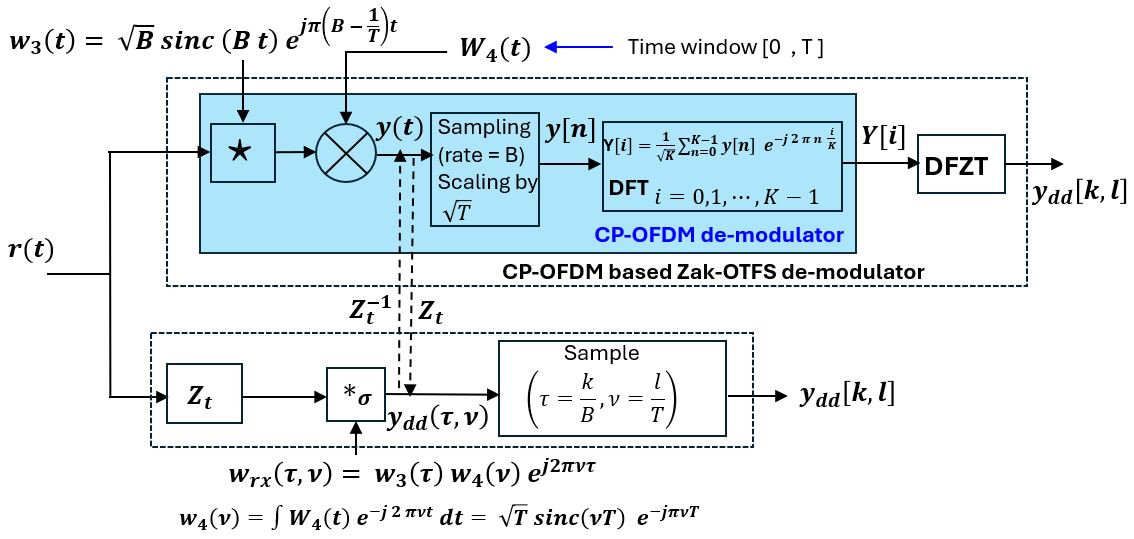}
\vspace{-2mm}
\caption{
Zak-OTFS de-modulation implemented as CP-OFDM de-modulator followed by DFZT (see top chain).}
\label{fig2}
\vspace{-1mm}
\end{figure*}
The transmitted signal $x(t)$ propagates through a linear time-varying (LTV) multipath channel, characterized by its DD spreading function $h_{\text{phy}}(\tau, \nu)$ given by
\begin{equation}
    h_{\text{phy}}(\tau, \nu) = \sum_{i=0}^{P-1} h_i \delta(\tau - \tau_i) \delta(\nu - \nu_i),
\end{equation}
where $h_i$, $\tau_i$, and $\nu_i$ represent the complex gain, delay, and Doppler shift of the $i$-th propagation path, respectively. The received time-domain signal is given by \cite{Bello, Hlawatsch2011}
\begin{equation}
\label{rteqn}
    r(t) = \iint h_{\text{phy}}(\tau,\nu) x(t-\tau) e^{j2\pi\nu(t-\tau)} \dd\tau \dd\nu + n(t),
\end{equation}
where $n(t)$ is additive white Gaussian noise (AWGN).
Let $\nu_{max}$ and $\tau_{max}$ denote the maximum possible path Doppler and Delay shift respectively.

The Zak-OTFS demodulator (receiver) can also be implemented as a post-processing of the output of the CP-OFDM demodulator as shown in Fig.~\ref{fig2} \cite{Zak_OTFS_over_CP_OFDM}.
This enables implementation of the Zak-OTFS de-modulator in existing 4G/5G modems. The output of the CP-OFDM demodulator is the signal $Y[i], i \in {\mathbb Z}$ received on the frequency-domain subcarriers ($Y[i]$ is received on the $i$-th subcarrier). The post-processor is the Discrete Frequency Zak Transform 
(DFZT) which takes as input $Y[i], i \in {\mathbb Z}$ and gives its discrete DD domain representation
    \begin{eqnarray}
\label{seqn12988}
    y_{dd}[k,l] & = & \text{DFZT}(Y[i]) \nonumber \\
    & = & \frac{1}{\sqrt{M}} \sum\limits_{p=0}^{M-1} Y[l + pN] \, e^{j 2 \pi (l + pN) \frac{k}{MN}},
\end{eqnarray}$k,l \in {\mathbb Z}$.
In the CP-OFDM demodulator, $r(t)$ is firstly time and bandwidth limited by filtering with a sinc filter $w_3(t)$ followed by a time-window $W_4(t)$ which limits the received time-domain signal to the time-interval $[0 \,,\, T]$ as in a CP-OFDM de-modulator we drop the cyclic prefix.
\begin{eqnarray}
    y(t) & = & W_4(t) \, \left[ w_3(t) \, \star \, r(t)\right], \nonumber \\
    w_3(t) & = & \sqrt{B} \, \text{sinc}(B \tau) \, e^{j \pi \left( B - \frac{1}{T}\right) t}, \nonumber \\
     W_4(t) & \Define  & \begin{cases}
       \frac{1}{\sqrt{T}} &, 0 \leq t < T, \\
       0  &, \mbox{\small{otherwise}} \\
    \end{cases}.
\end{eqnarray}The time and bandwidth limited received signal is then sampled resulting in the discrete-time
received signal
\begin{eqnarray}
    y[n] & = & \sqrt{T} y\left( t = \frac{n}{B}\right) .
\end{eqnarray}The frequency domain symbols received on the subcarriers is then given by the Discrete Fourier Transform (DFT) of $y[n]$, i.e.
\begin{eqnarray}
Y[i] = \frac{1}{\sqrt{MN}} \sum\limits_{n=0}^{MN-1} y[n] \, e^{-j 2 \pi n \frac{i}{MN}},
\end{eqnarray}$i \in {\mathbb Z}$.

In \cite{Zak_OTFS_over_CP_OFDM} it has been shown that $y_{dd}[k,l]$ is indeed the output of a Zak-OTFS demodulator, i.e., it is the match-filtered DD domain representation of the received signal sampled in the DD domain
\begin{eqnarray}
    y_{dd}[k,l] & = & y_{dd}\left( \tau = \frac{k}{B} \,,\, \nu = \frac{l}{T} \right), \nonumber \\
    y_{dd}(\tau, \nu) & = & w_{rx}(\tau, \nu) *_{\sigma} r_{dd}(\tau, \nu), \nonumber \\
    w_{rx}(\tau, \nu) & = & w_3(\tau) w_4(\nu) \, e^{j 2 \pi \nu \tau}, \nonumber \\
    w_4(\nu) & \hspace{-3mm} = & \hspace{-3mm} \int W_4(t) \, e^{-j 2 \pi \nu t} \, dt \, = \, \sqrt{T} \text{sinc}(\nu T) \, e^{-j \pi \nu T}. \nonumber \\
\end{eqnarray}Here $r_{dd}(\tau, \nu)$ is the DD domain representation of the received time-domain signal $r(t)$ and is given by its Zak transform \cite{Zak67,Janssen88,otfsbook}
\begin{eqnarray}
    r_{dd}(\tau, \nu) &  = & {\mathcal Z}_t\left( r(t) \right), \nonumber \\
    & = & \sqrt{\tau_p} \sum\limits_{n \in {\mathbb Z}} r(\tau + n \tau_p) \, e^{-j 2 \pi n \nu \tau_p}.
\end{eqnarray}
\begin{figure*}[h]
\vspace{-1mm}
\centering
\includegraphics[width=17.8cm, height=4.9cm]{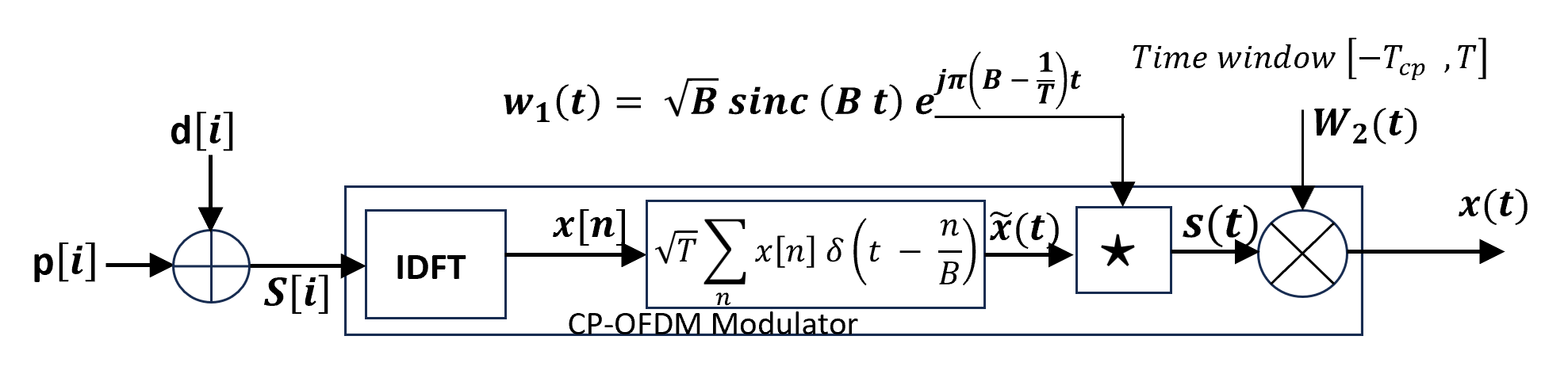}
\vspace{-2mm}
\caption{Block diagram of the ZC Spread Pilot CP-OFDM Transmitter. The DD domain ZC sequence is transformed via the IDFZT and superimposed on standard FD data symbols prior to CP-OFDM modulation.}
\label{fig:transmitter}
\vspace{-1mm}
\end{figure*}
The Zak-OTFS DD domain I/O relation between the input $x_{dd}[k,l]$ and the output $y_{dd}[k,l]$ is given by \cite{zakotfs2,otfsbook}
\begin{eqnarray}
\label{eqniorel1}
    y_{dd}[k,l] & = & h_{dd}[k,l] \, *_{\sigma} \, x_{dd}[k,l] \, + \, n_{dd}[k,l],
\end{eqnarray}where $n_{dd}[k,l]$ is the Zak transform of AWGN $n(t)$ match-filtered with $w_{rx}(\tau, \nu)$ followed by DD domain sampling, $*_{\sigma}$ denotes the discrete DD domain twisted convolution operation given by
\begin{eqnarray}
\label{eqniorel2}
    h_{dd}[k,l] \, *_{\sigma} \, x_{dd}[k,l] &  & \nonumber \\
    & & \hspace{-29mm} = \sum\limits_{k',l' \in {\mathbb Z}} h_{dd}[k',l'] \, x_{dd}[k - k', l - l'] \, e^{j 2 \pi l' \frac{(k - k')}{MN}}.
\end{eqnarray}Also, $h_{dd}[k,l]$ is the discrete DD domain effective channel filter given by
\begin{eqnarray}
\label{eqniorel3}
    h_{dd}[k,l] & \Define & h\left( \tau = \frac{k}{B} \,,\, \nu = \frac{l}{T}\right), \nonumber \\
    h_{dd}(\tau, \nu) & \Define & w_{rx}(\tau, \nu) \, *_{\sigma} \, h_{\text{phy}}(\tau, \nu) \, *_{\sigma} \, w_{tx}(\tau, \nu).
\end{eqnarray}
From this DD domain I/O relation between $y_{dd}[k,l]$ and $x_{dd}[k,l]$, the FD domain I/O relation between $S[i]$ and $Y[i]$ is given by the following theorem.

\begin{theorem}
    \label{thm1}
  The FD I/O relation between the input FD signal $S[i]$ and the output FD signal $Y[i]$ is given by
  \begin{eqnarray}
  \label{eqn2323}
    Y[i] & = &  \sum\limits_{l=0}^{MN-1}  S[i - l] \, H_f[i, l] \, + \, N[i],
  \end{eqnarray}where
  \begin{eqnarray}
  \label{eqn2864793}
  N[i] & \Define & \text{IDFZT}(n_{dd}[k,l]), \nonumber \\
  & = & \frac{1}{\sqrt{M}} \sum\limits_{k =0}^{MN-1} n_{dd}[k,i] \, e^{-j 2\pi \frac{i k}{MN}}, \nonumber \\
      H_f[i,l] & \Define & \sum\limits_{k=0}^{MN-1} h[k,l] \, e^{-j 2 \pi \frac{i k}{MN}}, \nonumber \\
      h[k,l] & \Define & \sum\limits_{n,m \in {\mathbb Z}} h[k + nMN, l + mMN].
  \end{eqnarray}
  Further the FD domain channel coefficients $H_f[i,l]$ are periodic in both $i$ and $l$ with period $MN$, i.e., for any $n,m \in {\mathbb Z}$
  \begin{eqnarray}
      H_f[i+nMN, l + mMN] & = H_f[i,l].
  \end{eqnarray}
\end{theorem}
\begin{IEEEproof}
See Appendix \ref{apa}.
\end{IEEEproof}

\section{Proposed DD domain sensing for CP-OFDM}
\label{embda}
\subsection{Transmitter signal processing}
The transmitter architecture is based on Zak-OTFS over CP-OFDM discussed in the previous section and is as shown in Fig.~\ref{fig:transmitter}.
Let $d[i], i=0,1,2,\cdots, K-1$, $K \Define MN$ denote the QAM information symbol transmitted on the $i$-th sub-carrier (each having unit average energy). These symbols could be the output of a forward error correction (FEC) encoder (e.g., Low Density Parity Check Code (LDPC)).
The overlaid sensing/pilot signal transmitted on the $i$-th sub-carrier is denoted by $p[i], i=0,1,2,\cdots, K-1$.
The symbol transmitted on the $i$-th subcarrier is then given by
\begin{eqnarray}
\label{eqn2874500}
    S[i] & = & \sqrt{E_d} \, d[i] \, + \, \sqrt{E_p} \, p[i],
\end{eqnarray}$i=0,1,\cdots, K-1$. Here $E_d$ is the average energy of the information symbol transmitted on each OFDM sub-carrier and $E_p$ is the pilot energy on each sub-carrier. Therefore, $E_p/E_d$ is the pilot to data power ratio (PDR). These $K$ symbols are then input to the CP-OFDM modulator whose output is the transmit signal $x(t)$ (see Fig.~\ref{fig:transmitter}).

The FD domain pilot signal $p[i]$ is a Zadoff-Chu (ZC) given by the Inverse Discrete Fourier Transform (IDFT) of its discrete-time representation \cite{}
\begin{eqnarray}
\label{xuneqn}
    X_u[n] & = & e^{- j 2 \pi u \frac{n (n+1)}{2 K}},
\end{eqnarray}where $u \in {\mathbb Z}$ is the root index, and is relatively prime to $K$. Hence
\begin{eqnarray}
    p[i] & = & \frac{1}{\sqrt{K}} \sum\limits_{n=0}^{K-1} X_u[n] \, e^{-j 2 \pi \frac{i n}{K}},
\end{eqnarray}$i=0,1,\cdots, K-1$.
ZC sequences are already being used as reference signals in existing 4G/5G modems.
The discrete DD domain representation of the ZC signal is given by the Discrete Frequency Zak Transform (DFZT) of $p[i]$, i.e. \cite{Mattu_ZC_DD}
\begin{eqnarray}
    X_{dd}^u[k,l] & = & \text{DFZT}(p[i]), \nonumber \\
    & = &   \frac{1}{\sqrt{M}} \sum\limits_{i=0}^{M-1} p[l + iN] \, e^{j 2 \pi (l + iN) \frac{k}{MN}}.
\end{eqnarray}$k,l \in {\mathbb Z}$.

\begin{figure*}[h]
\centering
\includegraphics[width=16.8cm, height=10.9cm]{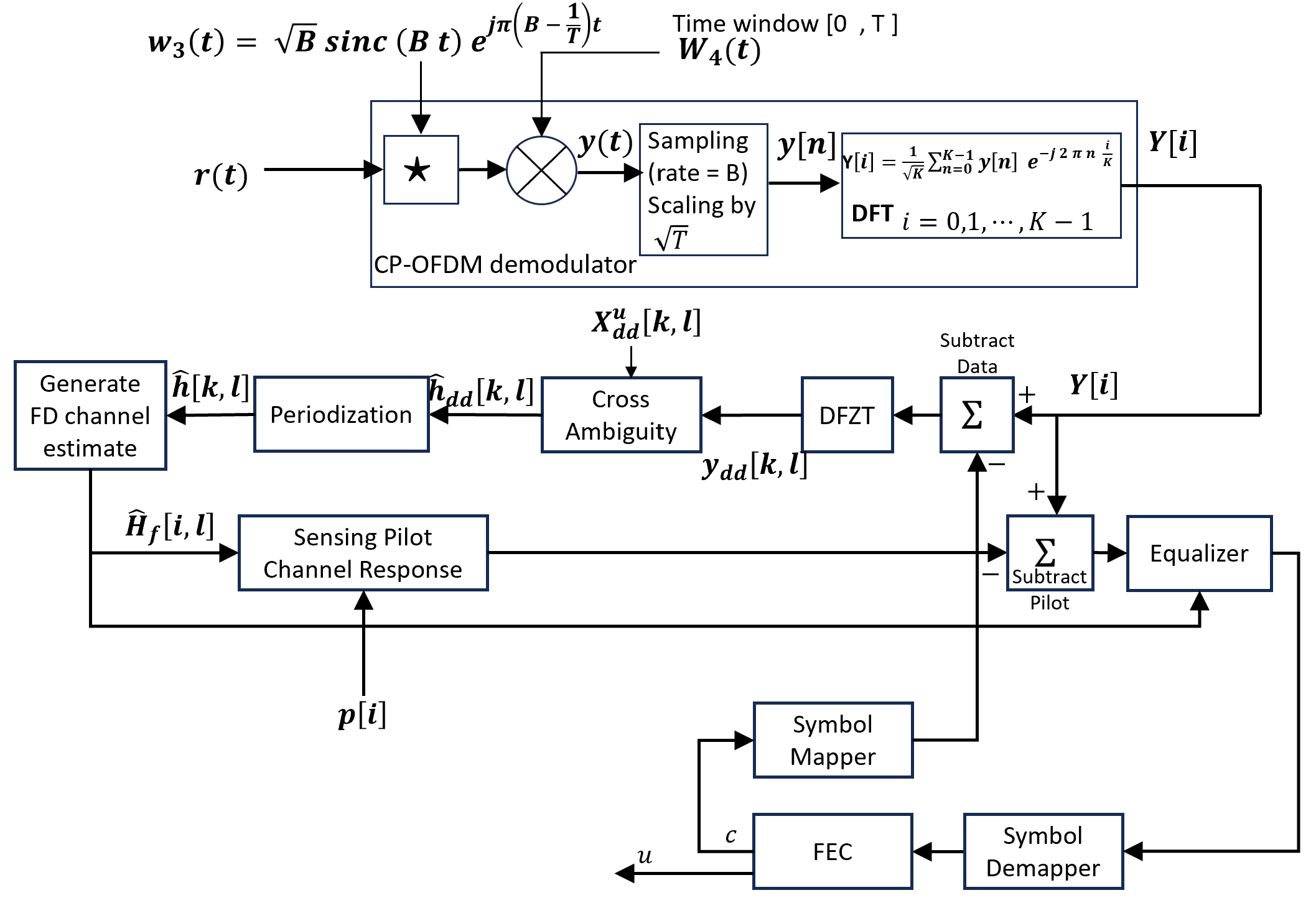}
\vspace{-2mm}
\caption{Block diagram of the DD domain sensing based CP-OFDM Receiver. }
\label{fig:receiver}
\vspace{-1mm}
\end{figure*}
\subsection{Receiver signal processing}
The receiver is depicted in Fig. \ref{fig:receiver}.
Let $Y[i], i=0,1,\cdots, K-1$ denote the FD symbols received on the $K$ CP-OFDM subcarriers at the output of the CP-OFDM demodulator (based on the Zak-OTFS over CP-OFDM demodulator discussed in the previous section).

\subsubsection{Receiver Step - $1$: Conversion to DD domain}
FD signal $Y[i]$ is firstly converted to its discrete DD domain representation through the Discrete Frequency Zak Transform (DFZT) given by
\begin{eqnarray}
       y_{dd}[k,l] & = & \text{DFZT}(Y[i]) \nonumber \\
    & = & \frac{1}{\sqrt{M}} \sum\limits_{p=0}^{M-1} Y[l + pN] \, e^{j 2 \pi (l + pN) \frac{k}{MN}}.
\end{eqnarray}Since $y_{dd}[k,l]$ is quasi-periodic, it suffices to compute it only for $k=0,1,\cdots,M-1$, $l=0,1,\cdots, N-1$ which has complexity $O(MN \log(MN))$. From (\ref{eqniorel1}) and (\ref{eqn2874500}), the DD domain I/O relation is given by

{\vspace{-4mm}
\small
\begin{eqnarray}
\label{eqn265423}
    y_{dd}[k,l] & \hspace{-3mm} = & \hspace{-3mm} h_{dd}[k,l] \, *_{\sigma} \, \left(  \sqrt{E_p} X_{dd}^u[k,l] \, + \, \sqrt{E_d} d_{dd}[k,l] \right) \nonumber \\
    & & \hspace{20mm} \, + \, n_{dd}[k,l], \nonumber \\
\end{eqnarray}\normalsize}where $d_{dd}[k,l]$ is the DD domain representation of the FD information symbols $d[i], i=0,1,\cdots, K-1$ and is given by their DFZT, i.e.
\begin{eqnarray}
    d_{dd}[k,l] & = & \frac{1}{\sqrt{M}} \sum\limits_{p=0}^{M-1} d[l + pN] \, e^{j 2 \pi (l + pN) \frac{k}{MN}}.
\end{eqnarray}

\subsubsection{Receiver Step - $2$: Channel acquisition in DD domain}
The receiver computes the cross ambiguity $A_{y,X^u}[k,l]$ between the received DD signal $y_{\text{dd}}[k,l]$ and the transmitted ZC pilot signal $X^u_{\text{dd}}[k,l]$
given by \cite{Interleaved_Pilots, Zak_OTFS_Identification}
\begin{eqnarray}
\label{eqnchest}
    A_{y,X^u}[k,l] \nonumber \\
    &  & \hspace{-26mm} = \hspace{-1.5mm} \sum_{k'=0}^{M-1} \sum_{l'=0}^{N-1} y_{\text{dd}}[k',l'] \left( X^u_{\text{dd}}[k'-k, l'-l] \right)^* e^{-j 2 \pi \frac{l(k'-k)}{K}},
\end{eqnarray}$(k,l) \in {\mathcal S}$, where ${\mathcal S}$ is the support set of the effective channel filter $h_{dd}[k,l]$ (i.e., set of DD taps $(k,l)$ where $h_{dd}[k,l]$ has significant energy).
Substituting the expression of $y_{dd}[k,l]$ from (\ref{eqn265423}) into (\ref{eqnchest}) we get\footnote{\footnotesize{Here we use the standard result from \cite{Zak_OTFS_Identification}, \cite{Spreadpaper} that the cross-ambiguity between a quasi-periodic signal $x_{dd}[k,l]$ and $y_{dd}[k,l] = h_{dd}[k,l] *_{\sigma} x_{dd}[k,l]$ is $A_{y,x}[k,l] = h_{dd}[k,l] *_{\sigma} A_{x,x}[k,l]$, where $A_{x,x} = \sum\limits_{k'=0}^{M-1}\sum\limits_{l'= 0}^{N-1} x_{dd}[k',l'] \, x_{dd}[k' - k, l' - l] \, e^{j 2 \pi l \frac{(k' - k)}{MN}}$ is the auto-ambiguity function of $x_{dd}[k,l]$.}}
\begin{eqnarray}
\label{eqn34}
    A_{y,X^u}[k,l] & = & \sqrt{E_p} \, h_{dd}[k,l] *_{\sigma} A_{X^u,X^u}[k,l] \nonumber \\
    & &  \hspace{5mm} + \sqrt{E_d} A_{d,X^u}[k,l] + A_{n,X^u}[k,l], 
\end{eqnarray}where $A_{X^u,X^u}[k,l]$ is the auto-ambiguity function of the pilot signal $X^u_{dd}[k,l]$, $A_{d, X^u}[k,l]$ is the cross-ambiguity function between the received communication/information signal $h_{dd}[k,l] *_{\sigma} d_{dd}[k,l]$ and the pilot signal $X^u_{dd}[k,l]$, and $A_{n, X^u}[k,l]$ is the cross-ambiguity between the noise $n_{dd}[k,l]$ and the pilot signal. The auto-ambiguity function for $X^u_{dd}[k,l]$ is 
\begin{eqnarray}
    A_{X^u, X^u}[k,l] & & \nonumber \\
    & & \hspace{-25mm} \Define \hspace{-1.5mm} \sum_{k'=0}^{M-1} \sum_{l'=0}^{N-1} X_{\text{dd}}^u[k',l'] \left( X^u_{\text{dd}}[k'-k, l'-l] \right)^* e^{-j 2 \pi \frac{l(k'-k)}{K}}.
\end{eqnarray}The auto-ambiguity function $A_{X^u, X^u}[k,l]$ is supported on the line $l = - u k \, \text{mod} \, MN$ \cite{Mattu_ZC_DD}. Note that the point $(k,l) = (0,0)$ lies on this line and $A_{X^u, X^u}[0,0] = MN$.
The auto-ambiguity function of the ZC sequence for $M = 289, N=5$, $u=7$ is shown in Fig.~\ref{fig:zc_ambiguity}. As can be seen, the auto-ambiguity is supported on the line $l=-uk \bmod MN$.
\begin{figure}[htbp]
\centering
\includegraphics[width=\columnwidth]{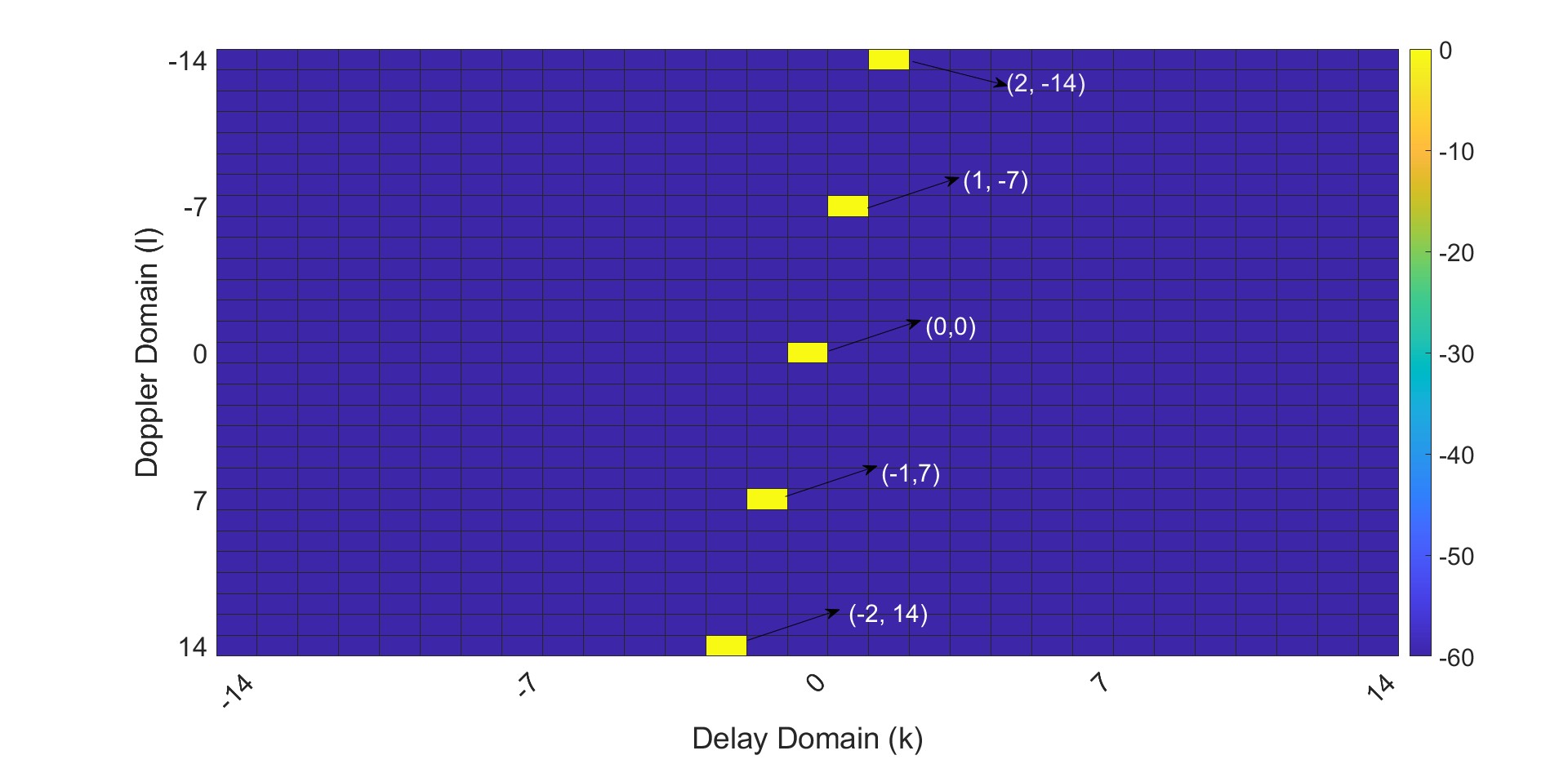}
\caption{Auto-ambiguity function $A_{X^u, X^u}[k,l]$ of the ZC sequence with $M = 289$, $N = 5$ and $u = 7$. The auto-ambiguity is strictly supported on the line $l=-uk \bmod MN$.}
\label{fig:zc_ambiguity}
\end{figure}

In (\ref{eqn34}), $h_{dd}[k,l] *_{\sigma} A_{X^u,X^u}[k,l]$ is the useful term for channel estimation. Expanding this term we get
\begin{eqnarray}
\label{eqn2864566}
h_{dd}[k,l] *_{\sigma} A_{X^u,X^u}[k,l] & = &  MN \, h_{dd}[k,l]  \nonumber \\
& & \hspace{-43mm} + \hspace{-6mm} \sum\limits_{\substack{(k', l') \ne (0,0) \\ l' = - u k' \, \text{mod} MN} } \hspace{-8mm} A_{X^u,X^u}[k',l'] \, h_{dd}[k - k' , l - l'] \, e^{j 2 \pi \frac{k'(l - l')}{MN}},
\end{eqnarray}where the first term is simply the effective DD domain channel filter and the second term is $h_{dd}[k,l]$ shifted by the DD points on the support of $A_{X^u, X^u}[k,l]$. We choose $u$ in such a way the support set of the first useful term (i.e., support ${\mathcal S}$ of $h_{dd}[k,l]$) does not overlap with the support of the second term and therefore
This is referred to as the channel being crystalline with respect to the support of the auto-ambiguity function $A_{X^u,X^u}[k,l]$.
\begin{figure}[htbp]
\centering
\includegraphics[width=\columnwidth]{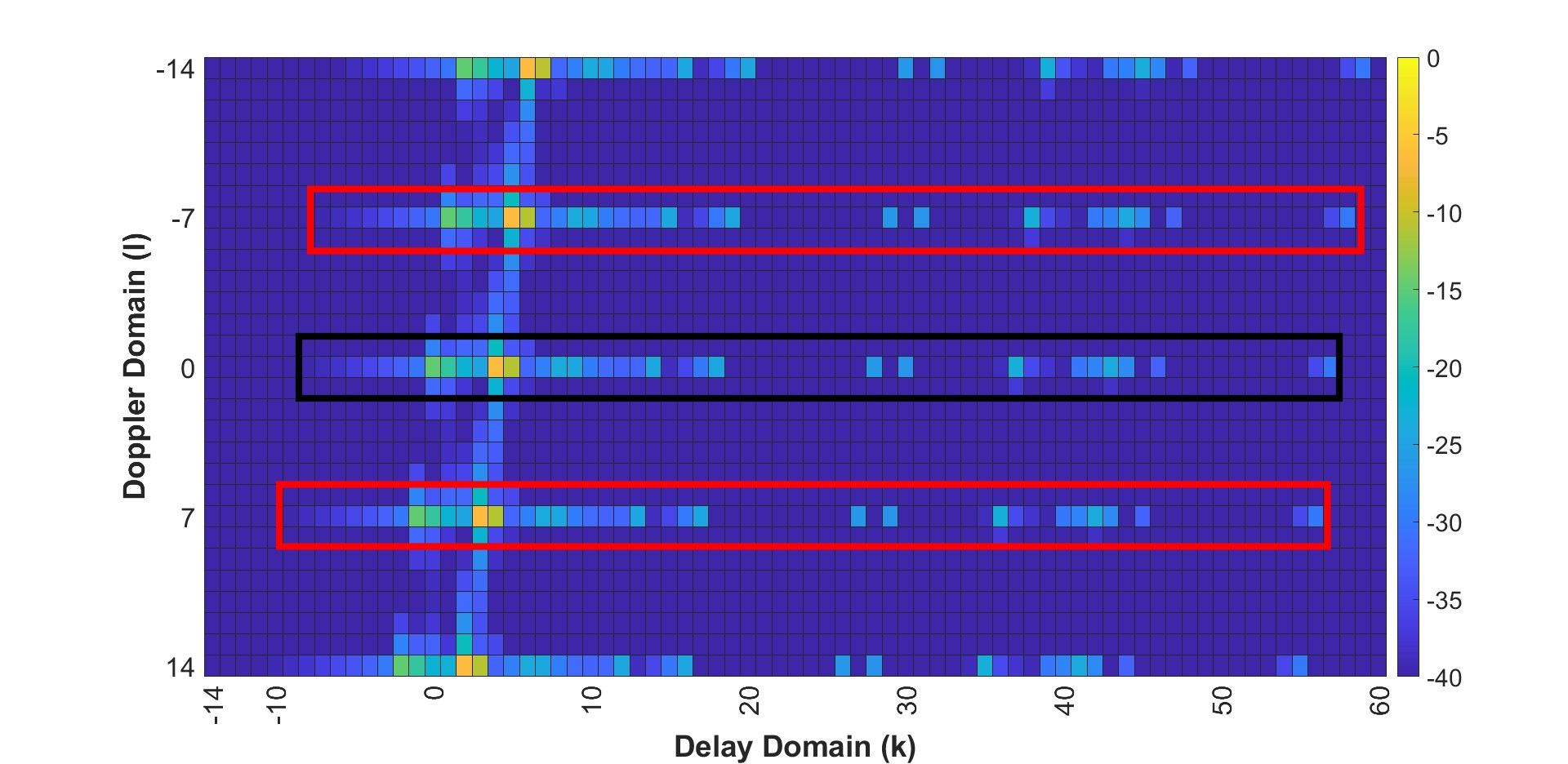}
\caption{Heatmap of $h_{dd}[k,l] *_{\sigma} A_{X^u,X^u}[k,l]$ with $M = 289$, $N = 5$ and $u = 7$.}
\label{fig:zc_crossambiguity}
\end{figure}
In Fig.~\ref{fig:zc_crossambiguity}
we plot the heatmap for $h_{dd}[k,l] *_{\sigma} A_{X^u,X^u}[k,l]$ for the 3GPP TDL-C channel model (non line-of-sight) \cite{3gpptr38901} where
$h_{\text{phy}}(\tau, \nu) = \sum\limits_{i=1}^{24} h_i \delta(\tau - \tau_i) \delta(\nu - \nu_i)$. Here $h_i$, $\tau_i$, $\nu_i$ is the complex gain, delay and Doppler shift of the $i$-th channel path. We model $\nu_i = \nu_{max} \cos(\theta_i),$ $i=1,2,\cdots, 24$ where $\nu_{max}$ is the maximum possible Doppler shift of any path and $\theta_i$, $i=1,2,\cdots, 24$ are i.i.d. uniformly in the interval $[0 \,,\, 2 \pi)$. In Fig.~\ref{fig:zc_crossambiguity} we choose $\nu_{max} = 2000$ Hz, i.e., a Doppler spread of $2 \nu_{max} = 4000$ Hz. Even for this high Doppler spread, it suffices to have $l_{max} = 1$, as is clear from the figure where the support set ${\mathcal S}$ (i.e., black rectangle) consists of DD taps with $l=-1, 0, 1$.
In (\ref{eqn2864566}), the first term is $MN h_{dd}[k,l]$, the support ${\mathcal S}$ for which is demarcated by a black rectangular boundary. The support for the other alias terms in (\ref{eqn2864566}) corresponding to the other points $(k',l') \ne (0,0)$ on the line $l = - u k \bmod MN$ is demarcated by red rectangles. Note that no rectangle overlaps with another rectangle, i.e., the channel is crystalline with respect to the auto-ambiguity function of the pilot signal.

Using (\ref{eqn2864566}) in (\ref{eqn34}) we finally have
\begin{eqnarray}
\label{eqn3737}
A_{y,X^u}[k,l] & = & \sqrt{E_p} \, MN \, h_{dd}[k,l] \nonumber \\
& & \hspace{-20mm} + \sqrt{E_p} \hspace{-8mm} \sum\limits_{\substack{(k', l') \ne (0,0) \\ l' = - u k' \, \text{mod} MN} } \hspace{-8mm} A_{X^u,X^u}[k',l'] \, h_{dd}[k - k' , l - l'] \, e^{j 2 \pi \frac{k'(l - l')}{MN}} \nonumber \\
    & & \hspace{-20mm} + \sqrt{E_d} A_{d,X^u}[k,l] + A_{n,X^u}[k,l], 
\end{eqnarray} and therefore an estimate of the effective DD domain channel filter is given by
\begin{eqnarray}
\label{eqn3838}
    \widehat{h}_{dd}[k,l] & \Define & \frac{A_{y, X^u}[k,l]}{MN \sqrt{E_p}},
\end{eqnarray}$(k,l) \in {\mathcal S}$.
From (\ref{eqnchest}) it follows that the complexity of computing $A_{y, X^u}[k,l]$ is therefore $O(\vert {\mathcal S}\vert MN)$,
where $\vert {\mathcal S}\vert$ is the cardinality of ${\mathcal S}$, i.e., the number of DD taps in the support set of $h_{dd}[k,l]$. 
In Section \ref{subsecnmse} we study the mean squared error of this estimate in greater detail.

\subsubsection{Receiver Step-$3$: Conversion of DD Domain Channel Estimate to Frequency Domain}
Using (\ref{eqn2864793}), an estimate of the frequency domain channel coefficients ${\widehat H}_{f}[i,l]$ is given by
\begin{eqnarray}
    {\widehat H}_f[i,l] & \Define & \sum\limits_{k' = 0}^{MN -1} \, {\widehat h}[k' , l] \, e^{-j 2 \pi \frac{i k'}{MN}}, \nonumber \\
    {\widehat h}[k' , l] & \Define & \sum\limits_{n,m \in {\mathbb Z}} {\widehat h}_{dd}[k' + nMN \,,\, l + mMN].
\end{eqnarray}Since the effective DD domain channel filter $h_{dd}[k,l]$ has significant energy for values of $(k,l) \in {\mathcal S}$, we compute ${\widehat H}_f[i,l]$, $i=0,1,\cdots, MN-1$ only for $l=-l_{max}, \cdots, l_{max}$ where the choice of $l_{max}$ depends on the Doppler spread of $h_{dd}[k,l]$ which in turn depends on $\nu_{max}$. We choose $l_{max} \leq \lceil T \nu_{max} \rceil$ such that the achieved NMSE does not improve appreciably on increasing it further. Note that $\widehat{H}_f[i,l]$ is also periodic in $l$ with period $MN$ and we set $\widehat{H}_f[i,l] = 0$ if $l_{max} < l \bmod MN < MN - l_{max}$, i.e.
\begin{eqnarray}
\label{eqn29746}
    {\widehat H}_f[i,l] & \hspace{-3mm} = & \hspace{-3mm}\begin{cases}
\sum\limits_{k' = 0}^{MN -1} \, {\widehat h}[k' , l] \, e^{-j 2 \pi \frac{i k'}{MN}} &, l \bmod MN \leq l_{max}\\
& \hspace{-15mm} \text{or}  \,\,  l \bmod MN \geq MN - l_{max}, \\
0 &, \text{otherwise}.
    \end{cases}
\end{eqnarray}Note that for a given $l$, $\widehat{H}_f[i,l]$, $l=0,1,\cdots, MN-1$ can be computed
through a $MN$-point Fast Fourier Transform (FFT) of $\widehat{h}[k,l]$, $k=0,1,\cdots, MN-1$ which has complexity only $O(MN\log (MN))$, and therefore the overall complexity of computing $\widehat{H}_f[i,l]$ is $O(l_{max}MN \log(MN))$.

\subsection{Receiver Step-$4$: Cancellation of ZC pilot in FD:}
From the I/O relation in (\ref{eqn2323}) of Theorem \ref{thm1}
it is clear that an estimate of the received pilot signal on the $i$-th sub-carrier is
\begin{eqnarray}
    \sum\limits_{l = - l_{max}}^{l_{max}} \hspace{-2mm} {\widehat H}_f[i,l] \, p[i - l],
\end{eqnarray}$i=0,1,\cdots, MN-1$.
Note that $H_f[i,l]$ is periodic in $l$ with period $MN$ (see Theorem \ref{thm1}) and therefore, the summation in the RHS of the I/O relation can be over any period of $MN$ values of $l$. In fact, for $l=-l_{max}, \cdots, -1$, $H_f[i,l] = H_f[i, l+MN]$, where $(l+MN)$ falls in the range $[0 \,,\, MN-1]$. 
This estimate is then canceled from the received FD signal resulting an almost ``data-only" signal
\begin{eqnarray}
\label{eqndataonly}
    Y_{data}[i] & \Define & Y[i] - \sqrt{E_p}\hspace{-2mm} \sum\limits_{l = - l_{max}}^{l_{max}} \hspace{-2mm} {\widehat H}_f[i,l] \, p[i - l],
\end{eqnarray}$i=0,1,\cdots, MN-1$.
The complexity of computing $Y_{data}[i]$, $i=0,1,\cdots, MN-1$ is therefore $O(l_{max} MN)$.

\subsection{Receiver Step-$5$: Frequency Domain Equalization and data detection}
From (\ref{eqn2323}) and (\ref{eqndataonly}), it follows that the data-only signal contains the useful data signal, residual pilot signal (since cancellation in step-$4$ is not perfect), channel estimation error signal and AWGN, i.e.
\begin{eqnarray}
\label{eqn1835}
 Y_{data}[i] & = &  \sqrt{E_d} \sum\limits_{l = - l_{max}}^{l_{max}}  \hspace{-2mm} \widehat{H}_f[i,l] \, d[i -l] \, + \, W[i],
\end{eqnarray}$i=0,1,\cdots, MN-1$, where
\begin{eqnarray}
    W[i] & \Define & \underbrace{\sqrt{E_d} \sum\limits_{l=0}^{MN-1} \hspace{-2mm} (H_f[i,l] - \widehat{H}_f[i,l]) d[i - l] }_{\text{Channel estimation error}}\nonumber \\
    & & \hspace{-20mm} + \, \underbrace{\sqrt{E_p} \sum\limits_{l=0}^{MN-1} \hspace{-2mm} (H_f[i,l] - \widehat{H}_f[i,l]) p[i - l]}_{\text{Residual pilot signal}} \, + \, N[i].
\end{eqnarray}
Organizing $Y_{data}[i], i=0,1,\cdots, MN-1$ into a vector ${\bf y}_{data} = (Y_{data}[0], Y_{data}[1], \cdots, Y_{data}[MN-1])^T \in {\mathbb C}^{MN \times 1}$, the I/O relation in (\ref{eqn1835}) is transformed into a matrix-vector form
\begin{eqnarray}
    {\bf y}_{data} & = & \sqrt{E_d} \, \widehat{\bf H} \, {\bf d} \, + \, {\bf w},
\end{eqnarray}where ${\bf d} = (d[0], d[1], \cdots, d[MN-1])^T$ is the vector of FD information symbols and ${\bf w} = (W[0], W[1], \cdots, W[MN-1])^T$. The element of the $MN \times MN$ matrix $\widehat{\bf H}$ in its $i$-th row and $l$-th column is ${\widehat H}_f[i,i-l]$. The matrix $\widehat{\bf H}$ is banded since $\widehat{H}_f[i,l]$ has non-zero values only for $l \bmod MN = 0 \cdots, l_{max}$ and $l \bmod MN = MN - l_{max} \cdots, MN -1$. Hence joint MMSE equalization of all $MN$ sub-carriers has complexity only $O(M^2N^2)$.
The equalized data symbols are then demapped to bits which are input to a FEC decoder (same as FEC used at the transmitter). The decoded bits are then mapped back to the data symbols denoted by
\begin{eqnarray}
    \widehat{d}[i], i=0,1,\cdots, MN-1.
\end{eqnarray}
\subsection{Iterative cancellation of received pilot and data}
The estimated information symbols $\widehat{d}[i]$, $i=0,1,\cdots, MN-1$ are periodized with period $MN$  is used to compute an estimate of the received data signal
\begin{eqnarray}
\label{eqn29748}
    \sqrt{E_d} \sum\limits_{l=-l_{max}}^{l_{max}} \widehat{H}_f[i,l] \, \widehat{d}[i-l],
\end{eqnarray}$i=0,1,\cdots, MN-1$. This is then subtracted from the
received signal resulting in the ``pilot-only" signal
\begin{eqnarray}
\label{eqn924801}
Y_p[i] & = & Y[i] - \sqrt{E_d} \sum\limits_{l=-l{max}}^{l_{max}} \widehat{H}_f[i,l] \, \widehat{d}[i-l].
\end{eqnarray}Computing $Y_p[i], i=0,1,\cdots, MN-1$ has complexity $O(l_{max} MN) $. This is then converted to a pilot only DD domain received signal (same as Step-$1$)
\begin{eqnarray}
    y_{p,dd}[k,l] & \Define & \text{DFZT}(Y_p[i]).
\end{eqnarray}A better estimate of the effective DD domain channel filter is then given by the cross-ambiguity between $y_{p,dd}[k,l]$ and $X^u_{dd}[k,l]$ (similar to Step-$2$), followed by conversion of the estimate to the frequency domain channel coefficients (see Step-$3$).
This is followed by estimation and cancellation of the received pilot signal resulting in a data-only signal from which the data symbols are detected (Step-$4$ and Step-$5$).
The detected information symbols are then used to estimate the received data signal (as in (\ref{eqn29748})) followed by its cancellation from the received signal (as in \ref{eqn924801}). These steps are repeated in an iterative manner resulting in improvement in NMSE and effective spectral efficiency with increasing number of iterations (see Fig.~\ref{fig:receiver}). Usually three to four iterations result in most of the improvement in performance with subsequent iterations resulting in diminishing improvements. 

\subsection{Receiver complexity}
The complexity of a traditional CP-OFDM receiver with per-carrier equalization (i.e., separate detection of data transmitted on each sub-carrier) is $O(MN \log (MN))$ (since there are $MN$ sub-carriers). In comparison, the complexity of the proposed DD domain sensing based CP-OFDM receiver is $O(M^2N^2)$ (steps $1-5$ and iterations between data and pilot cancellation/estimation). The proposed DD domain sensing based CP-OFDM receiver has higher complexity than a traditional CP-OFDM receiver, since it performs joint equalization of all $MN$ carriers instead of per-carrier equalization. Joint equalization improves robustness to channel Doppler spread and therefore the proposed DD domain sensing based CP-OFDM receiver achieves significantly better spectral efficiency than traditional CP-OFDM in high Doppler spread scenarios (see simulation results in Section \ref{simsec}).

Effectiveness of joint equalization in achieving robustness to channel Doppler spread is only because the proposed approach is able to estimate the frequency domain I/O relation $H_f[i,l]$ from an estimate of the effective DD domain channel filter $h_{dd}[k,l]$ through (\ref{eqn29746}). In traditional CP-OFDM it is not possible to estimate $H_f[i,l]$ for all $l$, since channel coefficient is estimated only on the sub-carrier where it is transmitted (i.e., $H_f[i,l]$ is estimated only for $l=0$ with $i$ being the index of pilot sub-carriers).

\subsection{On the accuracy of channel estimation in DD domain}
\label{subsecnmse}
Technically, we say that the channel is crystalline if the mean squared value of the aliasing term (second term in the RHS of (\ref{eqn2864566})) is smaller than the variance of the data and AWGN terms above, so that the estimation accuracy is not limited by the aliasing term, i.e. for all $(k,l) \in {\mathcal S}$
\begin{eqnarray}
\label{eqn3939}
    E_p {\Big \vert} \hspace{-7mm} \sum\limits_{\substack{(k', l') \ne (0,0) \\ l' = - u k' \, \text{mod} MN} } \hspace{-8mm} A_{X^u,X^u}[k',l'] \, h_{dd}[k - k' , l - l'] \, e^{j 2 \pi \frac{k'(l - l')}{MN}} {\Big \vert}^2 & & \nonumber \\
    & & \hspace{-75mm} \ll E_d \, {\mathbb E}\left[ \vert A_{d, X^u}[k,l] \vert^2 \right] \, + \, {\mathbb E}\left[ \vert A_{n, X^u}[k,l] \vert^2 \right].
\end{eqnarray}
When this condition is satisfied, we can estimate $h_{dd}[k,l]$ from the taps of $A_{y,X^u}[k,l]$ restricted to $(k,l) \in {\mathcal S}$. Next, we analyze the normalized mean squared estimation error (NMSE)

{\vspace{-4mm}
\small
\begin{eqnarray}
\label{eqn4040}
    \text{NMSE} & \Define & \frac{{\mathbb E}\left[ \sum\limits_{(k,l) \in {\mathcal S}}  \left\vert \widehat{h}_{dd}[k,l] -  h_{dd}[k,l] \right\vert^2\right]}{ {\mathbb E}\left[ \sum\limits_{(k,l) \in {\mathcal S}}  \left\vert h_{dd}[k,l] \right\vert^2\right]}.
\end{eqnarray}\normalsize}In the crystalline regime, from (\ref{eqn3737}), (\ref{eqn3838}) and (\ref{eqn3939}) it follows that

{\vspace{-4mm}
\small
\begin{eqnarray}
\label{eqn4141}
{\mathbb E}\left[ \sum\limits_{(k,l) \in {\mathcal S}}  \left\vert \widehat{h}_{dd}[k,l] -  h_{dd}[k,l] \right\vert^2\right] & & \nonumber \\
& & \hspace{-55mm} \approx \frac{E_d \sum\limits_{(k,l) \in {\mathcal S}} {\mathbb E}\left[ \vert A_{d, X^u}[k,l] \vert^2 \right]}{M^2N^2 E_p} \,  \, + \, \frac{\sum\limits_{(k,l) \in {\mathcal S}} {\mathbb E}\left[ \vert A_{n, X^u}[k,l] \vert^2 \right]}{M^2N^2E_p}.
\end{eqnarray}\normalsize}Since the data symbols $d[i], =0,1,\cdots, MN -1$ are i.i.d. unit energy and statistically independent of the sensing signal, we get (see Theorem $4$ in \cite{Spreadpaper})
\begin{eqnarray}
\label{eqn4242}
    \sum\limits_{(k,l) \in {\mathcal S}} {\mathbb E}\left[ \vert A_{d, X^u}[k,l] \vert^2 \right] & = & MN \vert {\mathcal S} \vert \hspace{-2mm} \sum\limits_{(k,l) \in {\mathcal S}} \hspace{-3mm} \vert h_{dd}[k,l] \vert^2,
\end{eqnarray}Similarly
\begin{eqnarray}
\label{eqn4343}
    \sum\limits_{(k,l) \in {\mathcal S}} {\mathbb E}\left[ \vert A_{n, X^u}[k,l] \vert^2 \right] & = & MN \, N_0 \, \vert {\mathcal S} \vert, \nonumber \\
    N_0 & \Define & {\mathbb E}\left[ \left\vert n_{dd}[k,l] \right\vert^2 \right].
\end{eqnarray}Using (\ref{eqn4242}) and (\ref{eqn4343}) in (\ref{eqn4141}) and (\ref{eqn4040}) we get
\begin{eqnarray}
\label{eqn4444}
\text{NMSE} & = & \frac{\vert {\mathcal S} \vert}{MN} \left( \frac{1}{\eta} + \frac{1}{\gamma_p} \right), \nonumber \\
\gamma_p & \Define & \frac{E_p \sum\limits_{(k,l) \in {\mathcal S}} \vert h_{dd}[k,l] \vert^2 }{N_0 } \,,\, 
\eta \, \Define \,  \frac{E_p}{E_d},
\end{eqnarray}where $\eta$ is the PDR and $\gamma_p$ is the pilot to noise ratio at the receiver. 
\begin{figure}[htbp]
\centering
\includegraphics[width=3.5in, height=2.0in]{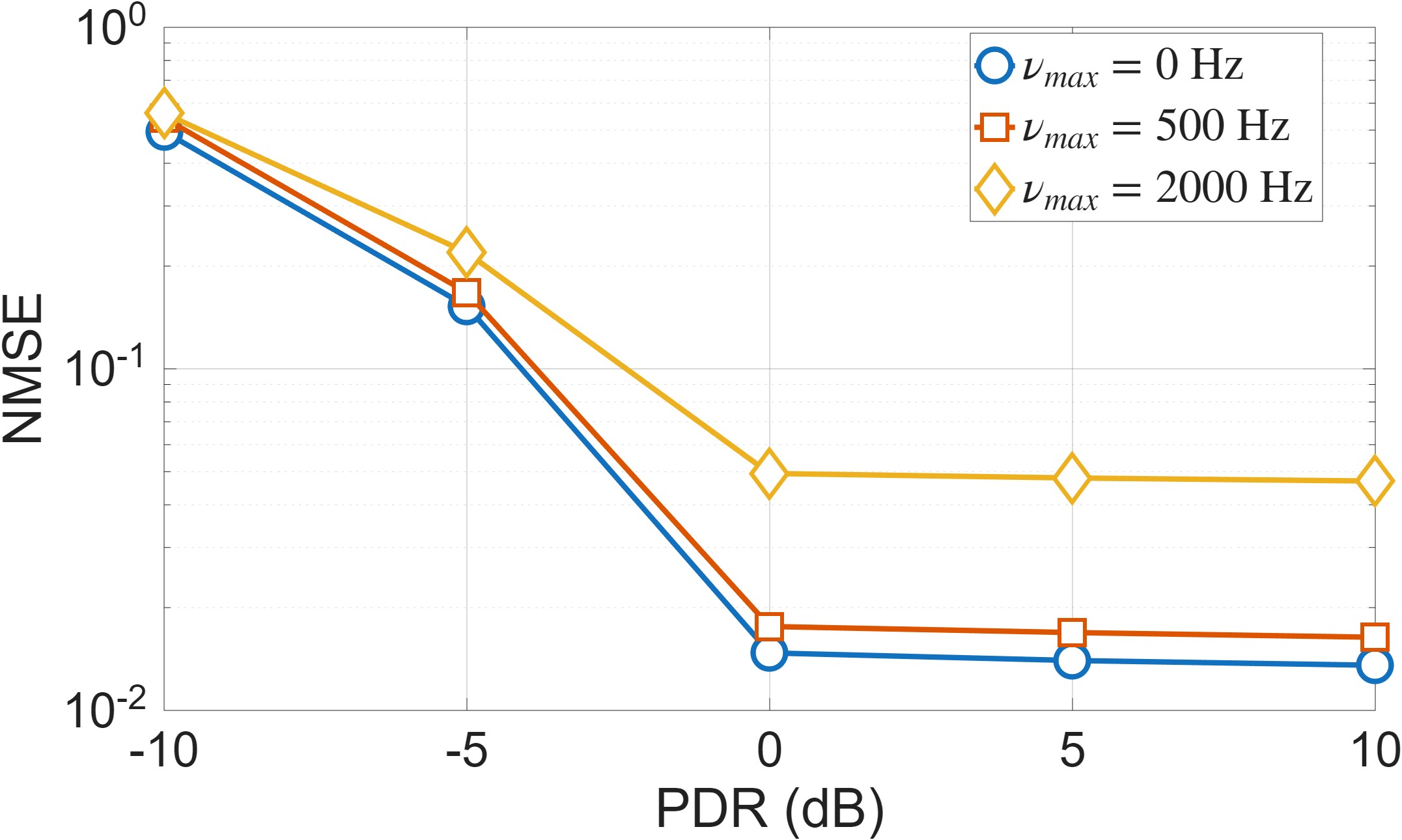}
\caption{NMSE vs PDR for a fixed total transmitted power (both pilot and data signal) to noise ratio of $12$ dB. 3GPP TDL-C (non-line-of-sight) channel model. $M = 289, N=5$, $u=7$.}
\label{fignmse}
\end{figure}
In Fig.~\ref{fignmse} we plot the NMSE for the 3GPP TDL-C (non-line-of-sight) channel model as considered in Fig.~\ref{fig:zc_crossambiguity}. We consider scaled channel path gains so that ${\mathbb E}[\sum\limits_{(k,l) \in {\mathcal S}} \vert h_{dd}[k,l] \vert^2] = 1 $.
Also, $M = 289, N=5$, $u=7$.
Fixed total transmit power to noise ratio (TPNR) is $12$ dB. 
TPNR is defined to be the total transmit power (both data and pilot combined) to the noise power, i.e., $(E_p + E_d)/N_0 = \gamma_p ( 1 + 1/\eta)$ since ${\mathbb E}[\sum\limits_{(k,l) \in {\mathcal S}} \vert h_{dd}[k,l] \vert^2] = 1 $.
The data signal to noise ratio (data SNR) denoted by $\gamma_d$ is defined to be $E_d/N_0 = \frac{E_p}{N_0} \frac{1}{\eta} = \gamma_p/\eta$.
\begin{eqnarray}
\label{eqn103659}
    \text{TPNR} & \Define & \gamma_p \left( 1 + \frac{1}{\eta} \right),\, \gamma_d \, \Define \, \frac{\gamma_p}{\eta}.
\end{eqnarray}In Fig.~\ref{fignmse},
for a given $\nu_{max}$, with increasing PDR $\eta$, NMSE decreases and saturates for high PDR. This is expected since $\text{TPNR}$ is fixed and from (\ref{eqn103659}) it follows that as $\eta \rightarrow \infty$, $\gamma_p \rightarrow TPNR$. Further, from (\ref{eqn4444}) it follows that as $\eta \rightarrow \infty$, NMSE converges to a fixed value $\text{NMSE} \rightarrow \frac{\vert {\mathcal S} \vert }{MN \, \text{TPNR}}$. Note that the simulated NMSE is higher than this lower bound NMSE expression since this expression assumes no aliasing (second term in (\ref{eqn2864566}) being zero), whereas due to channel Doppler spread there is always a finite amount of aliasing, and also due to sinc pulse shaping there is always a small fraction of the total channel energy $\sum\limits_{k,l \in {\mathbb Z}} \vert h_{dd}[k,l] \vert^2$ which leaks outside the finite support set ${\mathcal S}$.

For a given PDR $\eta$, NMSE increases with increasing maximum Doppler shift $\nu_{max}$. This is because, with increasing Doppler spread, both aliasing and the amount of channel energy outside ${\mathcal S}$ increases resulting in degradation in NMSE.

\section{Simulation Results}
\label{simsec}
We compare the spectral efficiency (SE) achieved with the proposed DD domain sensing based CP-OFDM to that achieved with traditional CP-OFDM.
In the proposed DD domain sensing based CP-OFDM we perform joint equalization of all sub-carriers, whereas we perform separate equalization of each sub-carrier in traditional CP-OFDM.

We consider the 3GPP TDL-C (non-line-of-sight) channel model as considered in Fig.~\ref{fig:zc_crossambiguity}.
For traditional CP-OFDM, we consider the 3GPP 5G NR numerology with $\Delta f = 15$ kHz and a total of $K = 1445$ subcarriers. The cyclic prefix duration is $T_{cp} = 4.7 \mu s$. For traditional CP-OFDM estimation, some of the carriers are used as pilot carrier (we only simulate one OFDM symbol). The channel estimate acquired on the pilot carriers is interpolated to obtain an estimate for the other data carriers. For both traditional CP-OFDM and the proposed DD domain sensing based CP-OFDM, the total bandwidth is $B = K \Delta f = 21.675$ MHz and duration is $(T +T_{cp}) = 71.36 \mu s$.

For the proposed DD domain sensing based CP-OFDM we consider the Zak-OTFS over CP-OFDM architecture as described in Section \ref{embda}. The Zak-OTFS parameters are $\nu_p = 75$ kHz, $\tau_p = 1/nu_p = 13.33 \mu s$, which implies $M = B \tau_p = 289, N= T \nu_p = 5$. In this approach, since we transmit a ZC pilot signal ($u=7$) overlaid on top of the CP-OFDM communication signal, we do not need pilot sub-carriers and therefore all $K = MN$ sub-carriers carry data symbols.

The information symbols in both methods are jointly coded using LDPC (FEC) as specified in 3GPP 5G NR \cite{3gppmcs2}. We compute the effective spectral efficiency for both traditional CP-OFDM and DD domain sensing based CP-OFDM as follows. 
For a given TPNR we optimize the PDR so as to be able to achieve the highest possible MCS combination (modulation and coding) with an average LDPC block error rate (BLER) at most $0.1$. Note that each MCS combination defines a choice of the modulation order and LDPC code rate \cite{3gppmcs1, 3gppmcs2}. For traditional CP-OFDM we also optimize the MCS w.r.t. the pilot arrangements in 3GPP 5G NR. The effective spectral efficiency (SE) is given by $N_I (1 - BLER)/(B (T + T_{cp}))$ bits/sec/Hz, where $N_I$ is the number of information bits corresponding to the highest achievable MCS.
\begin{figure}[htbp]
\hspace{-1mm}
\includegraphics[width=9.1cm, height=6cm]{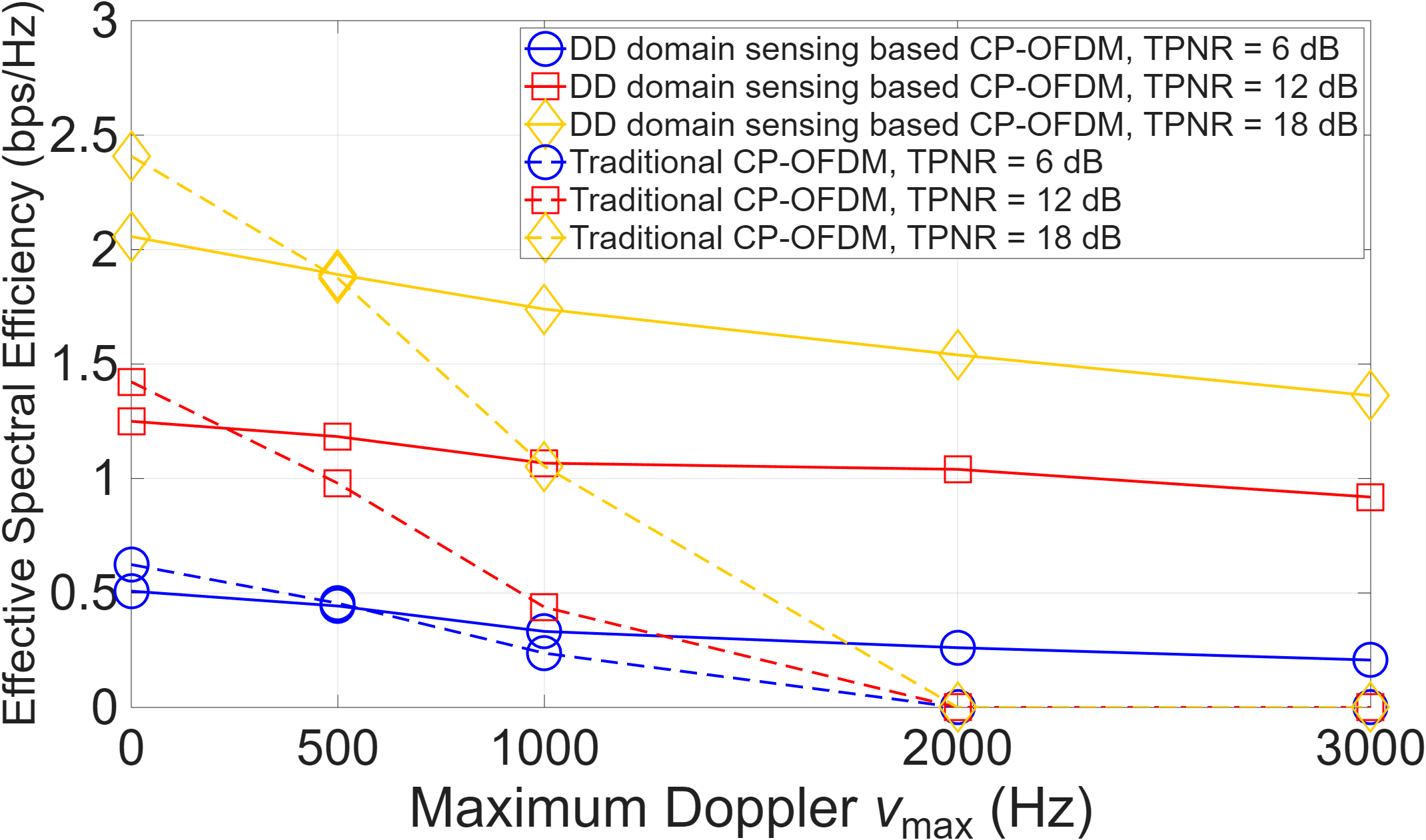}
\caption{Effective Spectral Efficiency (SE) versus Maximum Doppler Shift.}
\label{fig:se_doppler}
\end{figure}

In Fig.~\ref{fig:se_doppler} we plot the SE as a function of increasing $\nu_{max}$. It is observed that for low mobility scenarios ($\nu_{max} < 500$ Hz), traditional CP-OFDM achieves better SE than the proposed DD domain sensing based CP-OFDM. This is because, for small $\nu_{max}$ the ICI between sub-carriers is small and therefore SE performance of traditional CP-OFDM does not degrade. At the same time, the proposed DD domain sensing based CP-OFDM allocated lesser power to the data symbols as compared to traditional CP-OFDM since a significant fraction of the total power is used by the overlaid pilot signal (the optimal PDR which maximizes SE is around $0$ dB). However with increasing $\nu_{max}$,
the SE performance of traditional CP-OFDM degrades severely with no MCS being achievable (i.e., BLER greater than $0.1$) for $\nu_{max} \geq 2000$ Hz. At the same time, the SE performance of DD domain sensing based CP-OFDM degrades marginally with increasing $\nu_{max}$. The marginal degradation is due to the reduction in channel estimation accuracy with increasing $\nu_{max}$ (see Fig.~\ref{fignmse}). In high mobility scenarios, the proposed DD domain sensing based CP-OFDM achieves higher SE than traditional CP-OFDM since it equalizes the effect of inter-carrier interference through joint equalization of all sub-carriers, something which is made possible only because of the availability of accurate FD channel coefficients obtained through DD domain based channel estimation.
\begin{figure}[htbp]
\hspace{-1mm}
\includegraphics[width=9.7cm, height=6cm]{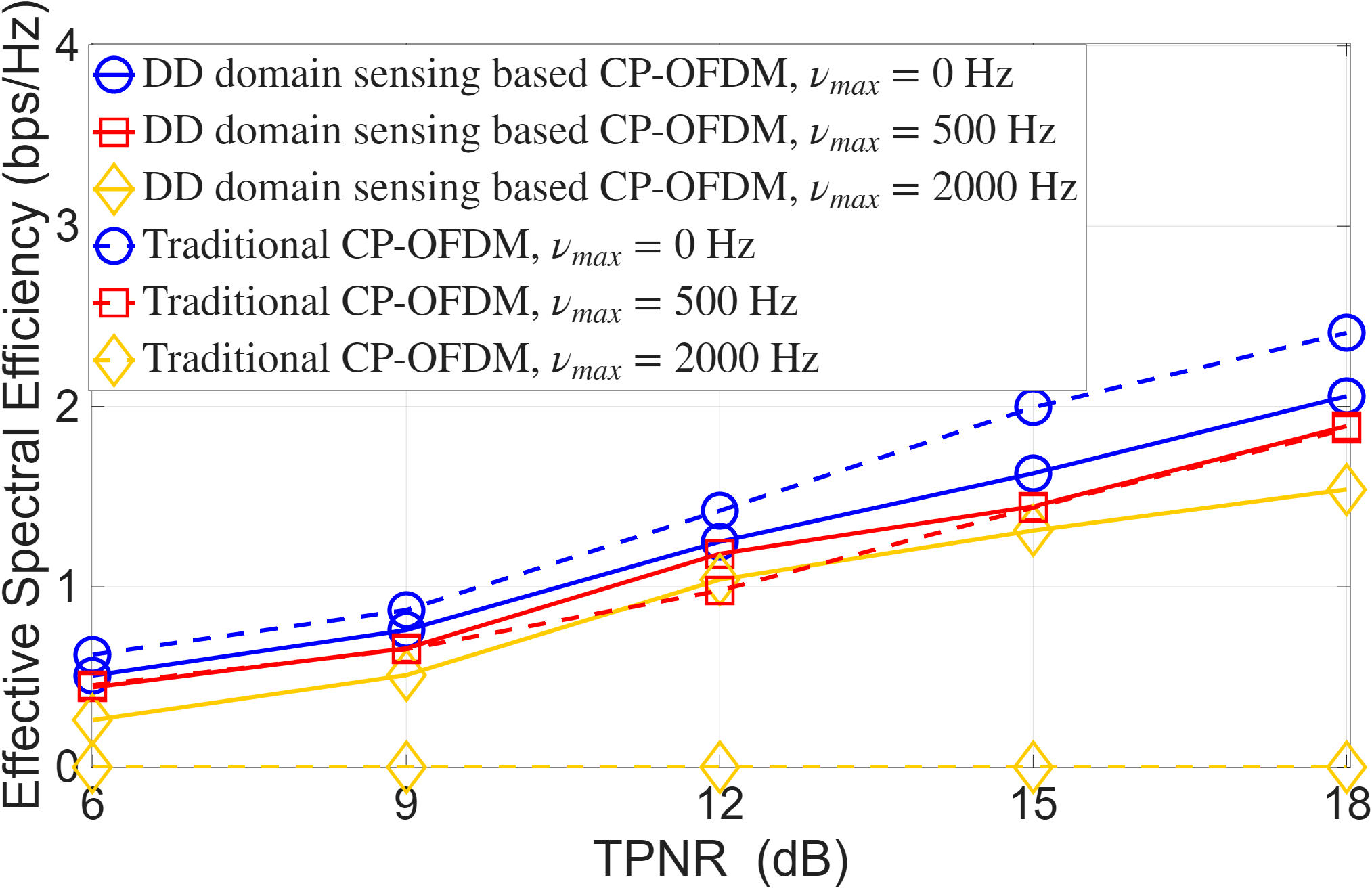}
\caption{Effective Spectral Efficiency (SE) versus TPNR.}
\label{fig:se_tpnr}
\end{figure}

In Fig.~\ref{fig:se_tpnr} we plot the SE as a function of increasing TPNR for a fixed $\nu_{max}$. For $\nu_{max} = 0$ Hz (no mobility), traditional CP-OFDM achieves higher SE compared to DD domain sensing based CP-OFDM. However, for higher mobility, the proposed DD domain sensing based CP-OFDM achieves better SE performance than traditional CP-OFDM.




\appendices

\section{Proof of Theorem \ref{thm1}}
\label{apa}
The Inverse Discrete Frequency Zak Transform (IDFZT) takes as input the DD domain representation of a signal and gives its discrete FD representation. In (\ref{sieqn12}), the sum in the R.H.S. can be over any delay period, i.e., for any $k' \in {\mathbb Z}$
\begin{eqnarray}
\label{eqn13846}
        S[i] & = & \text{IDFZT}(x_{dd}[k,l]) \nonumber \\
    & = & \frac{1}{\sqrt{M}} \sum\limits_{k=0}^{M-1} x_{dd}[k,i] \, e^{-j 2 \pi \frac{i k}{MN}}, \nonumber \\
    & = & \frac{1}{\sqrt{M}} \sum\limits_{k=k'}^{M-1+k'} x_{dd}[k,i] \, e^{-j 2 \pi \frac{i k}{MN}}.
\end{eqnarray}This is because, $x_{dd}[k,i] \, e^{-j 2 \pi \frac{i k}{MN}}$ is periodic in discrete delay index $k$ with period $M$, i.e., for any $n \in {\mathbb Z}$
\begin{eqnarray}
\label{eqn29845}
x_{dd}[(k+nM),i] \, e^{-j 2 \pi \frac{i (k+nM)}{MN}} & & \nonumber \\
& & \hspace{-31mm} \mya x_{dd}[k,i] \, e^{j 2 \pi \frac{i n}{N}} \, e^{-j 2 \pi \frac{i (k+nM)}{MN}}, \nonumber \\
& & \hspace{-31mm} = x_{dd}[k,i] \, e^{-j 2 \pi \frac{i k}{MN}},
\end{eqnarray}where step (a) follows from the
quasi-periodicity of $x_{dd}[k,l]$ in (\ref{qpeqn234}). Taking IDFZT of both sides of
the DD domain I/O relation in (\ref{eqniorel1}) we get
\begin{eqnarray}
\label{eqn28478}
    Y[i] & \mya & \text{IDFZT}\left( y_{dd}[k,l] \right), \nonumber \\
    & = &  \text{IDFZT}\left( h_{dd}[k,l] *_{\sigma} x_{dd}[k,] \right) \, + \,  N[i], \nonumber \\
    N[i] & \Define & \text{IDFZT}\left( n_{dd}[k,l] \right),
\end{eqnarray}where step (a) follows from the fact that $Y[i]$ is the FD representation of the received DD domain signal $y_{dd}[k,l]$ (see (\ref{seqn12988})). Also, $N[i]$ is the AWGN on the $i$-th subcarrier. Further from the expression for twisted convolution in (\ref{eqniorel2}) we get
\begin{eqnarray}
\label{eqni8264580}
    \text{IDFZT}\left(h_{dd}[k,l] \, *_{\sigma} \, x_{dd}[k,l]\right) &  & \nonumber \\
    & & \hspace{-50mm} = \hspace{-3mm} \sum\limits_{k',l' \in {\mathbb Z}} \hspace{-3mm} h_{dd}[k',l'] \, \text{IDFZT}\left( x_{dd}[k - k', l - l'] \, e^{j 2 \pi l' \frac{(k - k')}{MN}} \right).
\end{eqnarray}From (\ref{eqn13846}) it follows that
\begin{eqnarray}
    \text{IDFZT}\left( x_{dd}[k - k', l - l'] \, e^{j 2 \pi l' \frac{(k - k')}{MN}} \right) & & \nonumber \\
    & & \hspace{-60mm} = \frac{1}{\sqrt{M}}  \hspace{-3mm} \sum\limits_{k=k'}^{M-1+k'} \hspace{-5mm} x_{dd}[k - k', i - l'] \, e^{j 2 \pi l' \frac{(k - k')}{MN}} \, e^{-j 2 \pi \frac{i k}{MN}}, \nonumber \\
    & & \hspace{-60mm} = e^{-j 2 \pi \frac{i k'}{MN}} \, \underbrace{\frac{1}{\sqrt{M}}  \hspace{-1mm} \sum\limits_{k=0}^{M-1} \hspace{-1mm} x_{dd}[k, i - l'] \, e^{-j 2 \pi (i - l') \frac{k}{MN}}}_{= S[i - l']}, \nonumber \\
    & & \hspace{-60mm} = e^{-j 2 \pi \frac{i k'}{MN}}  \, S[ i - l'].
\end{eqnarray}Using this in (\ref{eqni8264580}) we then get

{\small
\vspace{-4mm}
\begin{eqnarray}
\label{eqn28747}
    \text{IDFZT}\left(h_{dd}[k,l] \, *_{\sigma} \, x_{dd}[k,l]\right) & \hspace{-3mm} = & \hspace{-4mm} \sum\limits_{k',l' \in {\mathbb Z}} \hspace{-3mm} e^{-j 2 \pi \frac{i k'}{MN}} \, h_{dd}[k',l'] \, S[i - l'],  \nonumber \\
    & & \hspace{-40mm} \mya \hspace{-1mm} \sum\limits_{k''=0}^{MN -1} \sum\limits_{l'' = 0}^{MN-1} \sum\limits_{n,m \in {\mathbb Z}} \hspace{-1mm} {\Big [} e^{-j 2 \pi \frac{i k''}{MN}} h_{dd}[k'' + nMN, l'' + mMN] \nonumber \\
    & & S[i- l'' - mMN] {\Big ]}, \nonumber \\
    & & \hspace{-40mm} = \hspace{-1mm}
    \sum\limits_{l''=0}^{MN-1} S[i - l''] \, H_f[i, l''], \nonumber \\
    & & \hspace{-45mm} H_f[i, l''] \Define 
    \hspace{-2mm} \sum\limits_{k''=0}^{MN-1} \hspace{-2mm} e^{-j 2 \pi \frac{i k''}{MN}} \, h[k'', l''], \nonumber \\
    & & \hspace{-45mm} h[k'', l''] \Define  \sum\limits_{n,m \in {\mathbb Z}} h_{dd}[k'' + nMN, l'' + mMN].
\end{eqnarray}\normalsize}where $H_f[\cdot, \cdot]$
is given by (\ref{eqn2864793}). Further, since $h[k,l]$ in (\ref{eqn2864793}) is periodic in both delay and Doppler with period $MN$, it follows that $H_f[i,l]$ is also $MN$-periodic in both $i$ and $l$.


\begin{thebibliography}{1}


\bibitem{IMT2030}
``Framework and overall objectives of the 
future development of IMT for 2030 and 
beyond," \emph{Recommendation ITU-R M.2160-0 }, International Telecommunication Union (ITU) - R, Nov. 2023.

\bibitem{OnTheRoadTo6G}
C. X. Wang et al., ``On the Road to 6G: Visions, Requirements, Key Technologies, and Testbeds,'' \emph{IEEE Communications Surveys \& Tutorials}, vol. 25, no. 2, pp. 905--974, 2023.



\bibitem{Nee2000} R.V.Nee, and R. Prasad, ``OFDM for Wireless Multimedia Communications,'' \emph{Artech House Inc.}, 2000.

\bibitem{Wang2006}  T.Wang, J. G.  Proakis, E . Masry and J. R.  Zeidler, ``Performance
degradation of OFDM systems due to Doppler spreading,'' \emph{IEEE Trans. on Wireless Commun.}, vol.5, no.6, June 2006.

\bibitem{Bello}
P.~A.~Bello, ``Characterization of Randomly Time-Variant Linear Channels," {\em IEEE Trans. Comm. Syst.}, vol. 11, pp. 360-393, 1963. 

\bibitem{Hlawatsch2011}
F. Hlawatsch and G. Matz, \emph{Wireless Communications Over Rapidly Time-Varying Channels}. New York, NY, USA: Academic Press, 2011.



\bibitem{zakotfs1}
S. K. Mohammed, R. Hadani, A. Chockalingam, and R. Calderbank, ``OTFS—A Mathematical Foundation for Communication and Radar Sensing in the Delay-Doppler Domain,'' \emph{IEEE BITS the Information Theory Magazine}, vol. 2, no. 2, pp. 36--55, 2022.

\bibitem{zakotfs2}
S. K. Mohammed, R. Hadani, A. Chockalingam, and R. Calderbank, ``OTFS—Predictability in the Delay-Doppler Domain and Its Value to Communication and Radar Sensing,"  \emph{IEEE BITS the Information Theory Magazine}, vol. 3, no. 2, pp. 7-31, June 2023.

\bibitem{otfsbook}
S. K. Mohammed, R. Hadani, and A. Chockalingam, \emph{OTFS Modulation: Theory and Applications}. Hoboken, NJ, USA: Wiley-IEEE Press, 2024.


 

\bibitem{Zak_OTFS_over_CP_OFDM}
S. K. Mohammed, S. Prakash, M. Ubadah, I. A. Khan, R. Hadani, S. Rakib, S. Kons, Y. Hebron, A. Chockalingam, and R. Calderbank, ``Zak-OTFS over CP-OFDM,'' \emph{arXiv preprint arXiv:2508.03906}, 2025.


\bibitem{Mattu_ZC_DD}
S. R. Mattu, I. A. Khan, V. Khammammetti, B. Dabak, S. K. Mohammed, K. Narayanan and R. Calderbank,``Multiple Preamble Detection with ZC Sequences in the Presence of Mobility and Delay Spread," \emph{2025 IEEE International Symposium on Information Theory (ISIT)}, Ann Arbor, MI, USA, 2025, pp. 1-6.

\bibitem{Zak67}
J. Zak, ``Finite translations in solid state physics,'' {\em Phy. Rev. Lett.}, 19, pp. 1385-1387, 1967.

\bibitem{Janssen88}
A. J. E. M. Janssen, ``The Zak transform: a signal transform for sampled time-continuous signals,'' {\em Philips J. Res.}, 43, pp. 23-69, 1988. 

\bibitem{Zak_OTFS_Identification}
D. Nisar, S. K. Mohammed, R. Hadani, A. Chockalingam, and R. Calderbank, ``Zak-OTFS for Identification of Linear Time-Varying Systems,'' \emph{IEEE Transactions on Signal Processing}, 2026.

\bibitem{Interleaved_Pilots}
J. Jayachandran, I. A. Khan, S. K. Mohammed, R. Hadani, A. Chockalingam, and R. Calderbank, ``Zak-OTFS With Interleaved Pilots to Extend the Region of Predictable Operation,'' \emph{IEEE Transactions on Vehicular Technology}, 2024.


\bibitem{Spreadpaper}
M. Ubadah, S. K. Mohammed, R. Hadani, S. Kons, A. Chockalingam and R. Calderbank, ``Zak-OTFS to Integrate Sensing the I/O Relation and Data Communication," \emph{arXiv preprint arXiv:2404.04182}, 2024.

\bibitem{3gpptr38901}
3GPP TR 38.901, ``Study on channel model for frequencies from 0.5 to 100 GHz," 3gpp Release 16, 2020.

\bibitem{3gppmcs1}
3GPP TS 38.211, ``NR; Physical Channels and Modulation," Release 15, 2018.

\bibitem{3gppmcs2}
3GPP TS 38.212, ``NR; Multiplexing and Channel Coding," Release 15, 2018.


\end{thebibliography}
\end{document}